\documentclass{bcp}
\usepackage[T1]{fontenc}
\usepackage{xcolor}
\newcommand{\revred}[1]{{\color{red}#1}}

\usepackage{lmodern,microtype,booktabs,array,etoolbox,needspace}
\usepackage[hidelinks]{hyperref}
\newtheorem{theorem}{Theorem}[section]
\newtheorem{proposition}[theorem]{Proposition}

\theoremstyle{definition}

\numberwithin{equation}{section}
\AtBeginEnvironment{thebibliography}{\interlinepenalty=10000}
\newcommand{\revisionbibitem}[1]{%
  \par\color{black}%
  \originalbibitem{#1}}
\AtBeginEnvironment{thebibliography}{%
  \let\originalbibitem\bibitem\let\bibitem\revisionbibitem}
\AtEndEnvironment{thebibliography}{\color{black}}
\newcommand{\g}{\mathfrak g}
\newcommand{\h}{\mathfrak h}
\newcommand{\C}{\mathbb C}
\newcommand{\R}{\mathbb R}
\newcommand{\Z}{\mathbb Z}
\newcommand{\N}{\mathbb N}
\newcommand{\dd}{\mathrm d}
\newcommand{\cL}{\mathcal L}
\newcommand{\cR}{\mathcal R}

\newcommand{\Ad}{\operatorname{Ad}}
\newcommand{\Der}{\operatorname{Der}}
\newcommand{\diver}{\operatorname{div}}

\newcommand{\SU}{\mathrm{SU}}
\newcommand{\SL}{\mathrm{SL}}
\newcommand{\SO}{\mathrm{SO}}
\newcommand{\U}{\mathrm U}

\newcommand{\CE}{\mathrm{CE}}
\newcommand{\ct}{\mathrm{ct}}
\newcommand{\DR}{\mathrm{dR}}
\newcommand{\cyc}{\mathrm{cyc}}

\newcommand{\tr}{\operatorname{tr}}

\makeatletter
\def\pag{\beg=\count0
  \def\phead{\noindent\vbox to4.5pt{\scriptsize\baselineskip9pt
    \centerline{FRONTIERS OF LIE THEORY}
    \centerline{BANACH CENTER PUBLICATIONS -- PROCEEDINGS CONTRIBUTION}
    \vss}\hfill}}
\patchcmd{\maketitlebcp}
 {The paper is in final form and no version of it will be published elsewhere.}
 {Prepared for Frontiers of Lie Theory: Computational Aspects and Applications,
 B\k{e}dlewo, 28 September--4 October 2025. Revised 10 October 2026.}
 {}{\PackageError{review}{BCP front-matter patch failed}{Check bcp.cls.}}
\makeatother
\hypersetup{pdftitle={Beyond the Circle: Current Algebras, Geometry and Higher Gauge Structures},
 pdfauthor={R. Campoamor-Stursberg, A. Marrani, M. Rausch de Traubenberg}}
\begin{document}
\keywords{current algebra, central extension, cyclic cohomology, harmonic analysis,
Kac--Moody algebra, soft geometry, free differential algebra.}
\mathclass{Primary 17B65; Secondary 17B56, 22E46, 58B15, 81R10.}
\abbrevauthors{R. Campoamor-Stursberg, A. Marrani and M. Rausch de Traubenberg}
\abbrevtitle{Current algebras beyond the circle}
\title{Beyond the circle:\\current algebras, geometry and higher gauge structures}
\author{Rutwig Campoamor-Stursberg}
\address{Instituto de Matem\'atica Interdisciplinar and Departamento de Geometr\'ia y Topolog\'ia,\\
Universidad Complutense de Madrid, Madrid, Spain\\
ORCID: 0000-0003-2907-8533\quad E-mail: rutwig@ucm.es}
\author{Alessio Marrani}
\address{Centre for Mathematics and Theoretical Physics,\\
University of Hertfordshire, Hatfield, United Kingdom\\
ORCID: 0000-0002-7597-1050\quad E-mail: a.marrani@herts.ac.uk}
\author{Michel Rausch de Traubenberg}
\address{Universit\'e de Strasbourg, CNRS, IPHC UMR7178, Strasbourg, France\\
ORCID: 0000-0001-5045-2353\quad E-mail: michel.rausch@iphc.cnrs.fr}
\maketitlebcp
\begin{abstract}
Replacing the circle by a manifold leads to current Lie algebras and their
central and derivation extensions. We separate the coefficient algebra,
harmonic analysis and geometric data. Differentials modulo exact forms describe
central extensions; divergence-free vector fields give explicit cocycles,
whose compatible derivations satisfy a stabilizer equation. Compact examples
include the three-sphere and its current--Virasoro extension. We also explain
the analytic qualifications for $\SL(2,\R)$ and its symmetric space and
the role of transported multiplication in soft geometries. The recent harmonic
constructions of Campoamor-Stursberg, Marrani and Rausch de Traubenberg are
compared with the cohomological work of Coquereaux, Frappat, Ragoucy and Sorba
and Castellani's loop geometry and infinite-dimensional free differential
algebras. A coefficient-algebra construction connects currents to higher
brackets and clarifies the additional requirements of a dynamical theory.
For Lie algebras of order $F$, central extensions obey coupled
binary and symmetric $F$-ary conditions, with explicit examples and obstructions.
\end{abstract}

\section{Introduction and scope}\label{s:intro}
The affine construction combines three ingredients: the algebra of functions on
$S^1$, an invariant form on a finite-dimensional Lie algebra, and differentiation
around the circle. Only the first of these is intrinsically one-dimensional.
The question addressed here is what survives when the circle is replaced by a
compact or noncompact manifold $M$. The most robust answer starts with
\begin{equation}\label{e:current}
 \cL(M,\g)=C^\infty(M,\g),\qquad [u,v](p)=[u(p),v(p)]_\g.
\end{equation}
Neither a metric nor a basis is needed for this bracket. A measure and a harmonic
basis make it computable; a cocycle determines a central extension; selected
derivations then determine its additional symmetries. Keeping these choices
separate is essential to a correct higher-dimensional analogue.

We use \emph{current algebra} for \eqref{e:current}. The expression
\emph{manifold Kac--Moody algebra}, used in the physical literature reviewed
below, denotes a specified central and derivation extension of such an algebra.
It does not assert that the result is a contragredient Kac--Moody algebra
defined by a generalized Cartan matrix \cite{Kac67,Moody67,Kac90}, or a
Borcherds algebra \cite{Borcherds88}. Likewise a Virasoro-type algebra of
vector fields along one chosen direction on $M$ is not automatically the
central extension of the full vector-field algebra $\mathfrak X(M)$.

There are three complementary historical strands. The algebraic description
of central extensions through differentials and cyclic homology was developed
in \cite{Kassel84,LodayQuillen84,Connes85}; toroidal constructions and their
presentations appear in \cite{MoritaYoshii85,MoodyRaoYokonuma90}.
Higher-dimensional smooth and super current algebras, including the
derivations compatible with a cocycle, were studied systematically in
\cite{CFRS90,RS92}, alongside surface and sphere constructions
\cite{FLRS90,Dowker91,BoseBruce90}. Second, Castellani introduced loop
geometry and its soft gauging \cite{Castellani88}, extended free differential
algebras (FDAs) to loop modes \cite{Castellani89}, and developed a classical
$D=4$, $N=1$ loop supergravity \cite{Castellani90}.
Third, explicit harmonic realizations were developed in
\cite{CMRcompact,CMRproc,CRfermions,CRvertex,CRspheres,CMRnoncompact,
CMRnoncompactproc,CMRsoft,CMRreview}.

The authorship and scope of this last sequence matter. The compact-manifold
foundation \cite{CMRcompact,CMRproc} is by Campoamor-Stursberg,
de Montigny and Rausch de Traubenberg; the free-field constructions
\cite{CRfermions,CRvertex} are by Campoamor-Stursberg and Rausch de
Traubenberg. The joint work with Marrani develops the noncompact
$\SL(2,\R)$ and $\SL(2,\R)/\U(1)$ cases \cite{CMRnoncompact}, their
proceedings account \cite{CMRnoncompactproc}, soft group manifolds
\cite{CMRsoft}, and the broader review \cite{CMRreview}. The present
contribution emphasizes their structural connections and the hypotheses needed
to interpret the mode formulae rigorously.

Sections~\ref{s:functions}--\ref{s:derivations} establish the algebraic and
cohomological framework. Sections~\ref{s:examples}--\ref{s:soft} treat explicit
examples and changes of geometry. Sections~\ref{s:super} and \ref{s:castellani}
compare the earlier superalgebraic and FDA approaches. Unless otherwise stated,
$\g$ is a finite-dimensional complex simple Lie algebra and $\kappa$ is a fixed
nonzero invariant symmetric bilinear form. Manifolds are smooth, Hausdorff and
second countable; the compact case has no boundary. Cohomology of smooth
current algebras is continuous cohomology. All displayed brackets use
$[T_a,T_b]=f_{ab}{}^cT_c$; factors of $i$ belong to explicitly chosen differential
operators, not to the abstract structure constants.

\Needspace{7\baselineskip}
\section{Coefficient algebras and harmonic realizations}\label{s:functions}
\subsection{An algebra must precede its Hilbert completion}
For a unital commutative associative complex algebra $A$, set
\begin{equation}\label{e:tensor}
 \cL_A=\g\otimes A,\qquad
 [x\otimes f,y\otimes g]=[x,y]\otimes fg.
\end{equation}
Associativity makes each term of the Jacobi identity proportional to $fgh$;
Jacobi therefore follows from that in $\g$. For $A=C^\infty(M)$ this is
\eqref{e:current}, since $\g$ is finite dimensional. On a compact manifold
the usual Fr\'echet topology makes multiplication and the bracket continuous.
For noncompact $M$, a useful default is $C_c^\infty(M)$ (smooth functions with compact support), with the inductive-limit topology of the Fr\'echet spaces of functions supported in a fixed compact set;
this is a nonunital algebra. Decaying coefficient algebras may also be used,
but their closure and continuity properties must be specified.

The Hilbert space $L^2(M,\mu)$ is generally \emph{not} a coefficient algebra.
For example, $f(x)=x^{-1/3}$ belongs to $L^2(0,1)$ whereas $f^2$ does not.
Multiplication $L^2\times L^2\to L^1$ follows from Cauchy--Schwarz, but it
does not produce an $L^2$ Lie bracket. Even if individual basis products are
square integrable, this does not establish closure for arbitrary Hilbert sums.
On compact $M$, $H^s(M)$ is the Sobolev space of distributions $f$ with $(1+\Delta)^{s/2}f\in L^2$, for a nonnegative elliptic Laplacian $\Delta$. For integer $s\ge0$, this means square-integrable weak derivatives through order $s$. It is a Banach algebra for $s>\dim M/2$, but a
first-order derivative normally lowers Sobolev order; $C^\infty(M)$ is the
natural common domain for the present derivation extensions.

Let $\mu$ be a positive smooth density, normalized to total mass one in the
compact case. We distinguish the Hermitian scalar product from the invariant
complex bilinear form:
\begin{equation}\label{e:pairings}
 (f,g)_\mu=\int_M\overline f g\,\mu,\qquad
 B(u,v)=\int_M\kappa(u,v)\,\mu.
\end{equation}
The latter is not an adjoint-trace ``Killing form'' of an infinite-dimensional
algebra. Its invariance is the pointwise identity
$\kappa([u,v],w)=\kappa(u,[v,w])$. Normalization involves a volume, not the
dimension of the underlying group.

\subsection{Compact groups and homogeneous spaces}
For a compact Lie group $K$, let $\widehat K$ denote its irreducible unitary
representations. With normalized Haar measure, Peter--Weyl theory gives
\begin{equation}\label{e:PW}
 L^2(K)\cong\widehat{\bigoplus}_{\pi\in\widehat K}V_\pi^*\otimes V_\pi,
 \qquad \psi^\pi_{ab}(k)=\sqrt{d_\pi}\,\pi(k)_{ab},\quad d_\pi=\dim V_\pi.
\end{equation}
The matrix coefficients $\psi^\pi_{ab}$, $1\le a,b\le d_\pi$, form an orthonormal basis \cite{PeterWeyl,PressleySegal}.
It is the \emph{irreducible} unitary representations that are finite dimensional;
general unitary representations of $K$ need not be. Their algebraic span
$\cR(K)$, the representative functions, is closed under multiplication because
tensor products of finite-dimensional representations decompose into finite
sums. Thus $\g\otimes\cR(K)$ is a genuine algebraic current algebra.

Writing a collective index $I=(\pi,a,b)$,
\begin{equation}\label{e:products}
 \psi_I\psi_J=\sum_K C_{IJ}^{K}\psi_K,\qquad
 C_{IJ}^{K}=\int_K\overline{\psi_K}\psi_I\psi_J\,\mu,
 \qquad \eta_{IJ}=\int_K\psi_I\psi_J\,\mu,
\end{equation}
the product coefficients are contractions of two Clebsch--Gordan intertwiners,
including their multiplicity indices. Consequently
\begin{equation}\label{e:modejac}
 [T_{aI},T_{bJ}]=f_{ab}{}^c\sum_K C_{IJ}^{K}T_{cK},\qquad
 \sum_P C_{IJ}^{P}C_{PK}^{Q}=\sum_P C_{JK}^{P}C_{IP}^{Q}.
\end{equation}
Here $T_{aI}=T_a\otimes\psi_I$, abbreviated $T_a\psi_I$.
The second equality is associativity, and is the coefficient-level reason for
Jacobi. For a general smooth basis these sums need not be finite; the
identities must then be interpreted in an appropriate convergent completion.
For example, eigenfunction expansions of an elliptic Laplacian on compact $M$
converge in $C^\infty$ for smooth functions, with rapid spectral decay.

If $H\subset K$ is closed, functions on $K/H$ are right $H$-invariant
functions on $K$, and
\begin{equation}\label{e:coset}
 L^2(K/H)\cong\widehat{\bigoplus}_{\pi\in\widehat K}
 V_\pi^*\otimes V_\pi^H.
\end{equation}
Here $V_\pi$ is the representation space of $\pi$, $V_\pi^*$ its dual, and $V_\pi^H=\{v\in V_\pi:\pi(h)v=v\text{ for every }h\in H\}$ is the subspace fixed by $H$. The hat over the sum denotes Hilbert completion.
The multiplicity is $\dim V_\pi^H$, not necessarily one. The subalgebra
$\cR(K)^H$ provides finite harmonic products. This avoids assuming a global
coordinate chart or an everywhere valid exponential parametrization.

On a compact semisimple group, minus the Killing form gives a bi-invariant
Riemannian metric. For noncompact real semisimple groups the Killing form is
indefinite; an auxiliary positive density or Riemannian metric must be chosen
separately. On a general homogeneous space a Riemannian metric requires an
$\Ad(H)$-invariant positive inner product on the tangent quotient. Left and
right translations are both isometries of a \emph{bi-invariant} metric;
this is not a statement about an arbitrary metric on a group manifold.

\section{Central extensions: differentials, cycles and cocycles}\label{s:central}
\subsection{The universal algebraic model}
Let $\Omega_A^1$ be the module of K\"ahler differentials, generated by $f\,\dd g$
with the linearity and Leibniz relations, and let $\overline\Omega_A^1=
\Omega_A^1/\dd A$. Bars denote classes in this quotient.
\begin{theorem}[Kassel's current-algebra extension]\label{t:universal}
For $A$ unital, commutative and associative over $\C$ and $\g$ simple, the
universal central extension of $\cL_A$ is
\begin{equation}\label{e:UCE}
 \widehat{\cL}_A=(\g\otimes A)\oplus\overline\Omega_A^1,\qquad
 [x\otimes f,y\otimes g]=[x,y]\otimes fg+
 \kappa(x,y)\overline{f\,\dd g}.
\end{equation}
The second summand is central. Every scalar central-extension class is
obtained by composing its cocycle with a linear functional
$\lambda:\overline\Omega_A^1\to\C$.
\end{theorem}
The universality statement is \cite{Kassel84}; see also
\cite{NeebWagemann}. The cocycle identity itself has a short direct proof.
Antisymmetry follows from $f\,\dd g+g\,\dd f=\dd(fg)$. In the cyclic sum
of the central terms, invariance of $\kappa$ leaves the common scalar
$\kappa([x,y],z)$ times
$\overline{fg\,\dd h+gh\,\dd f+hf\,\dd g}=
\overline{\dd(fgh)}=0$. Finally $\cL_A$ is perfect, since $\g$ is perfect
and $A$ has a unit, so the universal-extension assertion applies.

Equivalently, the scalar function
\begin{equation}\label{e:cyclic}
 \Phi(f,g)=\lambda(\overline{f\,\dd g})
\end{equation}
satisfies
$\Phi(f,g)=-\Phi(g,f)$ and
$\Phi(fg,h)+\Phi(gh,f)+\Phi(hf,g)=0$.
These are the degree-one cyclic cocycle identities. To fix the signs, the
Hochschild differential on $(q+1)$-linear cochains is
\begin{align}\label{e:hoch}
 (b\Phi)(a_0,\ldots,a_{q+1})
 &=\sum_{j=0}^{q}(-1)^j\Phi(a_0,\ldots,a_ja_{j+1},\ldots,a_{q+1})\notag\\
 &\quad+(-1)^{q+1}\Phi(a_{q+1}a_0,a_1,\ldots,a_q).
\end{align}
Cyclicity means $t\Phi=\Phi$, with
$(t\Phi)(a_0,\ldots,a_q)=(-1)^q\Phi(a_q,a_0,\ldots,a_{q-1})$.
Thus $b\Phi=0$ and $t\Phi=\Phi$ reproduce the two identities above
\cite{Connes85,LodayQuillen84}. This relation is more precise than calling
all antisymmetric bilinear forms Hochschild cocycles.

\subsection{The smooth category and closed currents}
For $A=C^\infty(M)$, continuous extensions use the \emph{smooth}
differentials $\Omega^1(M)$ and their topological quotient by
$\dd C^\infty(M)$. They should not be confused with the unrestricted
algebraic K\"ahler module of the underlying abstract algebra. In the compact
case the continuous analogue of Theorem~\ref{t:universal} gives
\begin{equation}\label{e:smoothcoho}
 H^2_\ct(\cL(M,\g),\C)
 \cong\bigl(\Omega^1(M)/\dd C^\infty(M)\bigr)'_\ct.
\end{equation}
$H^2_\ct$ denotes continuous Lie algebra cohomology: continuous antisymmetric bilinear cocycles modulo coboundaries of continuous linear functionals. Here $\Omega^1(M)$ is the Fr\'echet space of smooth one-forms, $\dd C^\infty(M)$ consists of exact one-forms, and the quotient carries its quotient topology. The prime with subscript $\ct$ denotes its continuous linear dual. Thus the isomorphism classifies scalar cocycles through continuous functionals on the universal central space.
The cocycle is $\Omega(u,v)=\overline{\kappa(u,\dd v)}$.
This follows from the continuous current and gauge algebra results of
\cite{Maier02,JanssensWockel}. Their compactly supported version also applies
to $C_c^\infty(M,\g)$, with $\Omega_c^1(M)/\dd C_c^\infty(M)$ when $M$ is noncompact. Here $\Omega_c^1$ means compactly supported smooth one-forms with their inductive-limit topology; the continuous dual of this quotient, rather than the quotient itself, again parametrizes scalar extensions.

A closed distributional one-current is a continuous linear functional $T$ on one-forms such that $T(\dd f)=0$ for every test function $f$. It therefore descends to the quotient by exact forms. For example, $T(\alpha)=\oint_\Gamma\alpha$ for a closed curve $\Gamma$ satisfies this condition because $\oint_\Gamma\dd f=0$. Its cocycle is $T(\kappa(u,\dd v))$. If $M$ is oriented of
dimension $n$ and $\gamma$ is a smooth closed $(n-1)$-form, a particularly
useful example is
\begin{equation}\label{e:gamma}
 \omega_\gamma(u,v)=\int_M\kappa(u,\dd v)\wedge\gamma.
\end{equation}
Here compactness, compact support or explicit decay eliminates boundary terms.
Distributional currents, such as integration over a closed curve in $M$, give
additional examples \cite[Chap.~4]{PressleySegal}; smooth $\gamma$ do not exhaust all continuous cocycles.

\begin{proposition}\label{p:div}
Let $\mu$ be a positive volume form and $L$ a smooth vector field. Then
\begin{equation}\label{e:omegaL}
 \omega_L(u,v)=\int_M\kappa(u,Lv)\,\mu
\end{equation}
is an antisymmetric Lie cocycle if $\diver_\mu L=0$. For nonzero $L$ it
represents a nonzero central-extension class.
\end{proposition}
\begin{proof}
Set $\gamma=\iota_L\mu$. Cartan's formula gives
$\dd\gamma=(\diver_\mu L)\mu$, while
$\dd g\wedge\iota_L\mu=(Lg)\mu$. Thus \eqref{e:omegaL} is
\eqref{e:gamma}; alternatively, antisymmetry and the cocycle identity follow
by integrating $L(fg)$ and $L(fgh)$. For nontriviality, $\Phi_L(f,g)=\int fLg\,\mu$ is nonzero when $L\ne0$: choose real $g$ of compact support with $Lg\ne0$ and $f=Lg$. For commutative coefficients, cyclic coboundaries $b\alpha(f,g)=\alpha(fg)-\alpha(gf)$ vanish \cite{RS92}. Directly on the Lie algebra, choose $h$ with $\kappa(h,h)\ne0$. The currents $hf,hg$ commute but $\omega_L(hf,hg)\ne0$, whereas every Lie algebra coboundary vanishes on commuting pairs.
\end{proof}

Divergence-free vector fields form a Lie algebra, since $\diver_\mu[V,W]=V(\diver_\mu W)-W(\diver_\mu V)$.
Note that $\gamma=\dd\beta$ does \emph{not} imply that
$\omega_\gamma$ is a Lie algebra coboundary. Indeed,
\begin{equation}\label{e:exactgamma}
 \int_M f\,\dd g\wedge\dd\beta
 =\int_M\dd f\wedge\dd g\wedge\beta
\end{equation}
The plus sign follows from $\dd(f\,\dd g\wedge\beta)=\dd f\wedge\dd g\wedge\beta-f\,\dd g\wedge\dd\beta$ and Stokes' theorem with zero boundary term. The integral need not vanish. The quotient in \eqref{e:smoothcoho} is
$\Omega^1/\dd\Omega^0$, not $H^1_\DR$; its numerator consists of all
one-forms, not only closed ones. In particular $S^3$ has
$H^2_\DR(S^3)=0$ ($H^k_\DR$ is the space of closed $k$-forms modulo exact ones) but admits nonzero cocycles \eqref{e:omegaL}.
For $n\ge2$, compactly supported choices of $\beta$ in disjoint coordinate
balls yield infinitely many independent cocycles. This observation explains
why the higher-dimensional central space is much larger than in the circle case.

For real forms whose complexification is simple, the same construction works
over $\R$. A complex simple algebra regarded as real has two independent
real invariant symmetric forms, obtained from the real and imaginary parts of
its complex form. A classification covering \emph{all} real simple algebras
must therefore use the universal space of invariant forms, as in
\cite{JanssensWockel}, rather than assume that every real invariant form is a
multiple of the real Killing form.

\section{Compatible derivations and selected central charges}\label{s:derivations}
A vector field $V$ differentiates \eqref{e:current}. After adjoining a central
generator $c$, we first ask for a lift with $[V,c]=0$ and no additional
central component in $[V,u]$. This specifies the meaning of compatibility.
\begin{proposition}[Compatibility equation]\label{p:compat}
For the cocycle \eqref{e:omegaL}, with $\diver_\mu L=0$, such a lift exists
if and only if
\begin{equation}\label{e:compat}
 [V,L]+(\diver_\mu V)L=0.
\end{equation}
Equivalently, $\cL_V(\iota_L\mu)=0$.
\end{proposition}
\begin{proof}
The mixed Jacobi identity is
$\omega_L(Vu,v)+\omega_L(u,Vv)=0$. Integration by parts expresses its
left-hand side as
\[
 -\int_M\kappa\bigl(u,([V,L]+(\diver_\mu V)L)v\bigr)\,\mu.
\]
Arbitrary test functions and nondegeneracy of $\kappa$ give
\eqref{e:compat}. The equivalent form follows from
$\cL_V\iota_L-\iota_L\cL_V=\iota_{[V,L]}$.
\end{proof}
This is the ungraded specialization of the compatibility analysis of
\cite{CFRS90,RS92}. Its solutions form a Lie algebra, being the stabilizer of
the form $\iota_L\mu$. They include every $V=fL$, since
$[fL,L]=-(Lf)L$ and $\diver_\mu(fL)=Lf$. They generally include more.
No restriction to Killing fields is necessary.

For several cocycles, impose one compatibility equation for each fixed
central generator. Let $D_i$, $1\leq i\leq r$, be commuting derivations of
$A$. A trace here is a linear functional $\tau:A\to\C$ with $\tau(fg)=\tau(gf)$, automatic for commutative $A$; it need not be a matrix trace or a multiplicative map. Assume $\tau(D_i f)=0$, giving $\tau((D_i f)g)=-\tau(fD_i g)$ by Leibniz. For smooth coefficients a standard example is $\tau(f)=\int_M f\,\mu$, with $D_i$ divergence-free. Functionals are continuous in the smooth category. Then
\begin{align}\label{e:multi}
 [x f,y g]&=[x,y]fg+\kappa(x,y)\sum_i\tau(fD_i g)c_i,\notag\\
 [d_i,xf]&=xD_i f,\qquad [d_i,d_j]=[d_i,c_j]=0
\end{align}
define a Lie algebra. The $c_i$ are \emph{central generators}; their scalar
eigenvalues in a representation are levels or central charges. They need not
span the full universal center. If $D_i\psi_I=\lambda_i(I)\psi_I$, the
central coefficient is
$\kappa(x,y)\sum_i\lambda_i(J)\eta_{IJ}c_i$.
Skewness follows from $(\lambda_i(I)+\lambda_i(J))\eta_{IJ}=0$.

For completeness, when $B$ in \eqref{e:pairings} is nondegenerate the extension
\eqref{e:multi} carries an invariant symmetric form extending $B$ by
$B(c_i,d_j)=-\delta_{ij}$, with all remaining new pairings zero.
The minus sign is forced by our convention $\omega_i(u,v)=B(u,D_i v)$.
This supplies the familiar pairing between a selected central direction and
its derivation, but is not a universal classification of maximal compatible
subalgebras.

There is a useful alternative to fixing the center. Every vector field $V$
acts on the universal central space by Lie derivative, with
\[
 \cL_V\Omega(u,v)=\Omega(Vu,v)+\Omega(u,Vv).
\]
Thus all vector fields act on the universal extension, with a generally
nontrivial action on its central space. After forming this semidirect product,
that space is abelian but need not be central in the enlarged algebra.
Lie derivative commutes with $\dd$ and preserves exact forms,
proving the quotient identity. For the scalar pushout $\omega=\lambda\circ\Omega$,
write $[u,v]_{\widehat{\cL}}=[u,v]+\omega(u,v)c$.
Let $V\in\Der(\cL)$ and $\chi:\cL\to\C$ be linear (continuous in the smooth case).
For $\widehat Vu=Vu+\chi(u)c$, $\widehat Vc=0$, compute
\begin{align*}
 \widehat V[u,v]_{\widehat{\cL}}
 &=V[u,v]+\chi([u,v])c,\\
 [\widehat Vu,v]_{\widehat{\cL}}+[u,\widehat Vv]_{\widehat{\cL}}
 &=[Vu,v]+[u,Vv]+\bigl(\omega(Vu,v)+\omega(u,Vv)\bigr)c.
\end{align*}
Since $V$ differentiates the base bracket, the lift exists exactly when
\begin{equation}\label{e:correctedlift}
 (\Delta_V\omega)(u,v):=\omega(Vu,v)+\omega(u,Vv)=\chi([u,v]).
\end{equation}
The right-hand bracket is in $\cL$. With
$(d_\CE\chi)(u,v)=-\chi([u,v])$, this is $\Delta_V\omega=-d_\CE\chi$:
invariance of the extension class, allowing a splitting change.

For the \emph{canonical} smooth current representative on compact $M$
and a vector field $V$, a correction does not relax invariance. Choose $h\in\g$ with
$\kappa(h,h)\ne0$ and put $u=hf$, $v=hg$. These currents commute, so
\eqref{e:correctedlift} gives
$\kappa(h,h)\lambda(\cL_V\overline{f\,\dd g})=0$.
These forms span the quotient by a partition of unity, so $\lambda\circ\cL_V=0$.
Conversely this gives $\Delta_V\omega=0$; perfectness then forces $\chi=0$.
In particular Proposition~\ref{p:compat} remains the criterion for
$\omega_L$. A non-invariant \emph{cohomologous} representative can still
require a correction: if $\omega'=\omega+d_\CE\alpha$ and
$\Delta_V\omega=0$, then
$\Delta_V\omega'=d_\CE(\alpha\circ V)$, and
$\chi=-\alpha\circ V$ satisfies \eqref{e:correctedlift}.

\section{Compact examples and a current--Virasoro construction}\label{s:examples}
\subsection{The circle and the torus}
For $A=\C[t,t^{-1}]$, the quotient $\overline\Omega_A^1$ is generated by
$\overline{t^{-1}\dd t}$: $t^r\dd t$ is exact unless $r=-1$.
Writing $x_m=x\otimes t^m$ gives
\begin{equation}\label{e:affine}
 [x_m,y_n]=[x,y]_{m+n}+n\delta_{m+n,0}\kappa(x,y)c,
 \qquad [d,x_m]=m x_m.
\end{equation}
Replacing $c$ by $-c$ gives the frequent convention with $m$ in the central
term. With $\ell_m=-t^{m+1}\partial_t$, one has
$[\ell_m,x_n]=-n x_{m+n}$ and
$[\ell_m,\ell_n]=(m-n)\ell_{m+n}$; their Virasoro extension has central
term $(m^3-m)\delta_{m+n,0}C/12$ \cite{PressleySegal,GoddardOlive,Fuks}.
These conventions will remain fixed below.

For $A=\C[t_1^{\pm1},\ldots,t_d^{\pm1}]$ and
$t^{\mathbf m}=t_1^{m_1}\cdots t_d^{m_d}$, define
\begin{equation}\label{e:torcenter}
 K_i(\mathbf r)=\overline{t^{\mathbf r}t_i^{-1}\dd t_i},\qquad
 \sum_i r_iK_i(\mathbf r)=0.
\end{equation}
The universal central term is
\begin{equation}\label{e:toruce}
 \kappa(x,y)\sum_i n_iK_i(\mathbf m+\mathbf n).
\end{equation}
Its homogeneous center has dimension $d$ in degree zero and $d-1$ in each
nonzero degree, hence is infinite dimensional for $d>1$. The familiar
$d$-center presentation
\begin{equation}\label{e:torfinite}
 [x_{\mathbf m},y_{\mathbf n}]=[x,y]_{\mathbf m+\mathbf n}
 +\kappa(x,y)\delta_{\mathbf m+\mathbf n,0}\sum_i n_i c_i
\end{equation}
is a quotient obtained by retaining only the degree-zero center. Toroidal
Chevalley-type presentations \cite{MoritaYoshii85,MoodyRaoYokonuma90} do not
identify every such quotient with an ordinary contragredient algebra.

\subsection{The two-sphere}
Use $\mu=\sin\theta\,\dd\theta\,\dd\varphi/(4\pi)$ and
$\psi_{\ell m}=\sqrt{4\pi}Y_{\ell m}$ in the Condon--Shortley convention.
Then $\psi_{00}=1$, $D=-i\partial_\varphi$ has eigenvalue $m$, and
\begin{align}\label{e:s2}
 \psi_{\ell m}\psi_{\ell'm'}
 &=\sum_{L,M}\sqrt{\frac{(2\ell+1)(2\ell'+1)}{2L+1}}
 C_{\ell0,\ell'0}^{L0}C_{\ell m,\ell'm'}^{LM}\psi_{LM},\notag\\
 \eta_{\ell m,\ell'm'}&=(-1)^m\delta_{\ell\ell'}\delta_{m+m^{\prime},0}.
\end{align}
Here $C_{\ell m,\ell'm'}^{LM}$ denotes an ordinary Clebsch--Gordan
coefficient, not a Wigner $3j$ symbol. Selection rules make the sum finite.
The central coefficient is $\kappa_{ab}m'\eta_{\ell m,\ell'm'}c$.
Although spherical coordinates degenerate at the poles, $\partial_\varphi$
is a globally smooth vector field with zeros there. Its use is therefore
geometrically legitimate. The vanishing of $H^1_\DR(S^2)$ does not eliminate
this current cocycle; Proposition~\ref{p:div} applies.

\subsection{The three-sphere as \texorpdfstring{$\SU(2)$}{SU(2)}}\label{ss:su2}
Let $W^j_{mn}=\sqrt{2j+1}\,D^j_{mn}$ be normalized matrix coefficients,
with $j\in\frac12\N_0$ and $m,n=-j,-j+1,\ldots,j$. Choose commuting
left and right infinitesimal circle generators so that
$D_LW^j_{mn}=mW^j_{mn}$ and $D_RW^j_{mn}=nW^j_{mn}$.
With standard compatible phases,
\begin{align}\label{e:SU2product}
 W^j_{mn}W^{j'}_{m'n'}&=
 \sum_{J=|j-j'|}^{j+j'}\sqrt{\frac{(2j+1)(2j'+1)}{2J+1}}
 C_{jm,j'm'}^{J,m+m'}C_{jn,j'n'}^{J,n+n'}W^J_{m+m',n+n'},\notag\\
 \overline{W^j_{mn}}&=(-1)^{m-n}W^j_{-m,-n},\qquad
 \eta_{IJ}=(-1)^{m-n}\delta_{jj'}\delta_{m+m^{\prime},0}\delta_{n+n^{\prime},0}.
\end{align}
The sum over $J$ advances in integer steps. Equations \eqref{e:SU2product}
determine all structure constants without a piecewise hypergeometric formula
or coordinate singularities. We write them as $C_{IJ}^{K}$.
The extension with two selected central directions is
\begin{align}\label{e:su2current}
 [T_{aI},T_{bJ}]&=f_{ab}{}^c\sum_K C_{IJ}^{K}T_{cK}
 +\kappa_{ab}(m'c_L+n'c_R)\eta_{IJ},\notag\\
 [d_L,T_{aI}]&=mT_{aI},\qquad [d_R,T_{aI}]=nT_{aI}.
\end{align}
Here $I=(j,m,n)$ and $T_{aI}=T_a\otimes W^j_{mn}$.
Both $D_L$ and $D_R$ preserve Haar measure and commute. Equation
\eqref{e:multi} proves all Jacobi identities. The full smooth universal
extension has a much larger center than $\C c_L\oplus\C c_R$.

The compatible vector fields are not exhausted by these two generators.
If $a,b$ are smooth functions invariant under both circle actions, then
\begin{equation}\label{e:compatiblecounter}
 V=aD_L+bD_R
\end{equation}
has zero divergence and commutes with both $D_L,D_R$. Such functions include
arbitrary smooth functions of $|z_1|^2$ in
$S^3=\{(z_1,z_2):|z_1|^2+|z_2|^2=1\}$. This gives an
infinite-dimensional family compatible with both cocycles. The two constant
generators provide a useful diagonal choice, not a uniqueness theorem.

Let $\h\subset\g$ be a Cartan subalgebra. Relative to
$\h\otimes1$, $d_L$ and $d_R$, a root vector $E_\alpha W^j_{mn}$ has
weight $(\alpha,m,n)$. For fixed admissible $m,n$, infinitely many $j$
occur. The weight spaces therefore have infinite multiplicity, and the
centralizer also contains $\h W^j_{00}$ for every integral $j$.
This finite-dimensional diagonal subalgebra is not a self-centralizing Cartan
subalgebra of the full current algebra. These weight decompositions do not
by themselves provide simple roots or a Chevalley--Serre presentation.

\subsection{A Virasoro cocycle along a chosen direction}\label{ss:vir}
The following elementary formulation makes both the scope and the signs of
the construction transparent. Let $A$ be commutative, $D$ a derivation and
$\tau$ a linear functional such that $\tau(Df)=0$. Regard $A$ as a Lie
algebra with bracket $\{f,g\}=gDf-fDg$, represented by
$\ell_f=-fD$ whenever this representation is faithful.
\begin{proposition}\label{p:vir}
The bilinear form
\begin{equation}\label{e:vircoc}
 \nu(f,g)=\frac1{12}\tau\bigl((D^3f-Df)g\bigr)
\end{equation}
is an antisymmetric Lie cocycle. Its linear-derivative part is a coboundary.
Moreover $\ell_f$ acts on the current algebra by $\ell_f(xg)=-xfDg$ and
preserves the current cocycle $\kappa(x,y)\tau(gDh)$.
\end{proposition}
\begin{proof}
Integration by parts for $\tau$ proves antisymmetry. For the cubic part,
move two derivatives in each cyclic term to obtain
\[
 \sum_\cyc\tau(D^3\bigl(\{f,g\}\bigr)\,h)
 =\tau\!\sum_\cyc
 \bigl(gD^2f-fD^2g\bigr)D^2h=0.
\]
The six terms cancel pairwise by commutativity. For the linear part, integration by parts gives $\tau(\{f,g\})=2\tau((Df)g)$. Fix the convention $(d_\CE\alpha)(f,g)=-\alpha(\{f,g\})$. With $\alpha(f)=\tau(f)/24$, the linear term in \eqref{e:vircoc} is exactly $d_\CE\alpha$. It can be removed by the change of central lifts $\ell_f^{\rm new}=\ell_f-\tau(f)C/24$, leaving the cubic cocycle. Finally
\[
 \tau\bigl((-fDg)Dh+gD(-fDh)\bigr)
 =-\tau\bigl(D(fgDh)\bigr)=0.
\]
This is the mixed Jacobi identity for one vector field $\ell_f$ and two currents $xg,yh$: $\ell_f$ acts on both coefficient functions and fixes the current center.
\end{proof}

For $A=\cR(\SU(2))$, $\tau=\int\mu$ and $D=D_L$, this yields
\begin{align}\label{e:su2vir}
 [\ell_I,\ell_J]&=(m-m')\sum_K C_{IJ}^{K}\ell_K
 +\frac{m^3-m}{12}\eta_{IJ}C,\notag\\
 [\ell_I,T_{aJ}]&=-m'\sum_K C_{IJ}^{K}T_{aK},\qquad
 [d_R,\ell_I]=n\ell_I.
\end{align}
Together with \eqref{e:su2current} retaining $c_L$ and omitting $c_R$,
these relations define the claimed current--Virasoro semidirect extension.
Here $d_L=-\ell_1$ is already in the vector-field algebra; it is not
adjoined a second time. In particular $[d_L,\ell_I]=m\ell_I$.
The phase $(-1)^{m-n}$ in $\eta_{IJ}$ belongs in the Virasoro term as
well as in the current term; it cannot be discarded in this fixed basis.
The action of $d_R$ on the $\ell_I$ must also be included. The $c_R$
cocycle is not preserved by all $fD_L$, so one cannot keep both current
centers while adjoining all these vector fields with trivial center action.

Proposition~\ref{p:vir} constructs a central extension; it does not classify all central extensions of this vector-field algebra. The factor $1/12$ and linear subtraction recover the circle expression $(m^3-m)/12$ for $D=t\partial_t$, $\tau(t^m)=\delta_{m,0}$. On $\SU(2)$ half-integral weights also occur. A global Laurent subalgebra would require an invertible smooth function $t$ with $Dt=t$; a local fibre angle does not supply one, since the Hopf circle bundle $S^3\to S^2$ has no global trivialization. The free-field realizations on $S^2,T^2$ \cite{CRfermions,CRvertex,CRspheres} use particular representations and regularizations; a Sugawara stress tensor on another manifold needs a separate construction.

\section{Noncompact harmonic analysis}\label{s:noncompact}
For a unimodular type-I Lie group $G$, the Plancherel transform has the form
\begin{equation}\label{e:plancherel}
 L^2(G)\cong\int_{\widehat G}^{\oplus}
 \operatorname{HS}(\mathcal H_\pi)\,\dd\mu_{\rm Pl}(\pi).
\end{equation}
$\widehat G$ is the unitary dual (equivalence classes of irreducible unitary representations), $\mathcal H_\pi$ is the Hilbert space on which $\pi$ acts, and $\operatorname{HS}(\mathcal H_\pi)$ consists of operators $A$ with $\tr(A^*A)<\infty$. The direct integral is the Hilbert space of measurable fields $A_\pi$ with $\int\|A_\pi\|_{\rm HS}^2\dd\mu_{\rm Pl}<\infty$; $\mu_{\rm Pl}$ is the Plancherel measure. Unimodular means that left and right Haar measures agree. Type~I is the representation-theoretic regularity that permits a decomposition into irreducible representations by such a measurable direct integral.
The spectral measure, rather than the mere existence of irreducible unitary
representations, determines which representations occur. For
$G=\SL(2,\R)$ the regular representation contains discrete series and
principal continuous series; complementary series and the trivial
representation do not carry Plancherel mass
\cite{Bargmann,HarishChandra,Knapp,CMRnoncompact}.
For connected real semisimple groups with finite center, existence of
discrete series is equivalent to existence of a compact Cartan subgroup
(equality of the complex ranks of $G$ and a maximal compact subgroup).
This is not an equality involving the real split rank.

The principal-series matrix coefficients used in the transform are
generalized eigenfunctions, not individual vectors of an orthonormal
Hilbert basis. Their delta-normalization and the Plancherel measure must
remain explicit. For $\SL(2,\R)/\SO(2)$, right-$\SO(2)$ invariance selects
the spherical part; ordinary discrete series have no fixed vector and drop
out. This is the hyperbolic plane, whose invariant positive metric is
different from the indefinite Killing metric on $\SL(2,\R)$ itself.

Campoamor-Stursberg, Marrani and Rausch de Traubenberg
\cite{CMRnoncompact,CMRnoncompactproc} give both a spectral construction
and a countable orthonormal realization, the latter using the basis they
call the Losert basis. This is an explicit improvement over merely knowing
that $L^2(G)$ is separable. A countable basis can span a representation
whose Casimir has continuous spectrum: its vectors need not diagonalize
that Casimir. Changing to this basis therefore creates neither a new
irreducible discrete series nor a new unitary dual.

For rigor, first form $\cL_c(G,\g)=\g\otimes C_c^\infty(G)$, its cocycles,
and its differentiations. Left and right invariant vector fields preserve
this test space and Haar measure. Consequently \eqref{e:omegaL} converges
and Proposition~\ref{p:compat} applies without a boundary at infinity.
One may then express these operations through \eqref{e:plancherel};
spectral brackets are distributional kernels tested against the transformed
test functions. In a countable basis the same operations have discrete
matrix coefficients, with convergence controlled by the chosen test
topology. A unitary map of Hilbert spaces alone does not justify
termwise multiplication or differentiation of arbitrary $\ell^2$ sums.

Larger Schwartz-type spaces are valuable, but one must verify that
multiplication, the chosen derivatives and the cocycle integrals are
continuous there. Nor does a formal integration-by-parts identity prove
self-adjointness: a symmetric operator requires a domain and an extension
analysis. For a complete measure-preserving real flow, pullback gives a
strongly continuous unitary group, and its self-adjoint generator is supplied
by Stone's theorem. This is the appropriate operator-theoretic statement
behind the Hermitian differential operators in the harmonic constructions.

\section{Soft geometry: what is transported?}\label{s:soft}
In the group-geometric approach, a coframe $\mu^A$ replaces the
Maurer--Cartan forms of a Lie algebra with structure constants
$c_{BC}{}^A$. Its soft curvature is
\begin{equation}\label{e:softcurvature}
 R^A=\dd\mu^A+\tfrac12c_{BC}{}^A\mu^B\wedge\mu^C,
 \qquad \dd R^A+c_{BC}{}^A\mu^B\wedge R^C=0.
\end{equation}
The second identity follows from the first and Jacobi. This curvature
measures departure from the Maurer--Cartan equations; it is not in general
the Riemann curvature of a Levi--Civita connection. Nonzero $R^A$ permits
additional geometry, but supplies no field equations by itself
\cite{Castellani88,Castellani90,Castellani92,CDF}.

Suppose the underlying compact smooth manifold is fixed and two positive
densities satisfy $\mu'=\rho\mu$, with $\rho\in C^\infty(M)$ strictly
positive. The map
\begin{equation}\label{e:unitary}
 U:L^2(M,\mu)\longrightarrow L^2(M,\mu'),\qquad Uf=s f,
 \quad s=\rho^{-1/2},
\end{equation}
is unitary, and transports an orthonormal basis. It does not preserve
ordinary multiplication. For a vector-field derivation $D$,
\begin{equation}\label{e:transportD}
 D^U=UDU^{-1}=D+\tfrac12D(\log\rho),
 \qquad D^U1=\tfrac12D(\log\rho).
\end{equation}
The last term is a multiplication operator, and
\begin{equation}\label{e:leibnizdefect}
 D^U(fg)-(D^Uf)g-f(D^Ug)=-\tfrac12D(\log\rho)fg.
\end{equation}
Thus $D^U$ is a derivation for the ordinary product exactly when
$D\rho=0$. Similarity preserves operator commutators, but not the
Leibniz rule for a product it does not transport.

\begin{proposition}[Two consistent constructions]\label{p:soft}
There are two distinct ways to combine \eqref{e:unitary} with current algebras.
\emph{(a)} Keep the ordinary product and use actual derivations (vector-field operators $D$ satisfying Leibniz) and closed
currents for the chosen geometry.
\emph{(b)} Transport the product as well:
\begin{equation}\label{e:star}
 f\star g=U(U^{-1}f\,U^{-1}g)=s^{-1}fg.
\end{equation}
Then $(C^\infty(M),\star)$ has unit $s$, $D^U$ is a derivation, and
transporting the bracket and cocycles gives an algebra isomorphic to the
original extension.
\end{proposition}
\begin{proof}
In (b), associativity, Leibniz and Jacobi follow by applying $U$ to the
corresponding identities. In particular
$\omega^U(Uu,Uv)=\omega(u,v)$ is a cocycle for the transported bracket.
In (a), if $\diver_\mu L=0$, then
\begin{equation}\label{e:densityL}
 L'=\rho^{-1}L,\qquad
 \diver_{\mu'}L'=0,\qquad \omega_{\mu',L'}=\omega_{\mu,L}.
\end{equation}
These equalities follow from
$\iota_{L'}\mu'=\iota_L\mu$. Compatibility is then checked using
Proposition~\ref{p:compat}, not inferred from similarity.
\end{proof}

This distinction refines, rather than replaces, the discussion in
\cite[Sec.~2.3]{CMRsoft}. That paper writes both the ordinary product of
the rescaled modes and a deformed tensor product, and explicitly retains
the original $D_i$ and original integration measure in the ordinary-product
central extension. Write $\psi^{\prime}_I=s\psi_I=\sum_JP_I{}^J\psi_J$ and $D_i\psi_J=\lambda_i(J)\psi_J$. Denote the matrix of $D_i$ in the original basis by $\mathsf D_i$, with $(\mathsf D_i)_J{}^K=\lambda_i(J)\delta_J{}^K$.
\[ D_i\psi^{\prime}_I=\sum_K(P\mathsf D_iP^{-1})_I{}^K\psi^{\prime}_K. \] The first index labels the basis function acted upon, the second its output coefficient; this specifies the index convention. This is the matrix of the same operator $D_i$ in a new basis, not the conjugated differential operator $D_i^U=UD_iU^{-1}$. The latter is diagonal on $\psi^{\prime}_I$ but satisfies Leibniz only for \eqref{e:star}, unless $D_i\rho=0$. Infinite matrices are understood on the smooth domain.
For generic smooth $s$, the algebraic span of the rescaled representative
functions need not be closed under ordinary multiplication; a smooth
completion is then necessary.

Metric deformation alone does not change $C^\infty(M,\g)$ as an abstract
smooth current algebra. It changes distinguished harmonic bases, spectra,
isometries and possibly the chosen cocycles. For example, a Berger metric
$\sigma_1^2+\sigma_2^2+\varepsilon^2\sigma_3^2$ on $\SU(2)$ has volume
form $\varepsilon$ times the round one; after normalization its Haar
measure is unchanged, although its Laplace spectrum and isometry group
change. Claims of a new algebra must therefore specify which product,
topology, central quotient and geometric decorations are being compared.

\section{The cohomological legacy of the Sorba collaborations}\label{s:super}
The two papers \cite{CFRS90,RS92} are structural predecessors of the harmonic
approach, not just examples of sphere algebras. Coquereaux, Frappat, Ragoucy
and Sorba \cite{CFRS90} study
\begin{equation}\label{e:supercircle}
 A_N=C^\infty(S^1)\otimes\Lambda(\theta^1,\ldots,\theta^N),
 \qquad \g\otimes A_N,
\end{equation}
where $\g$ is ordinary and the parity comes from the coefficient algebra.
They determine cyclic cocycles through Grassmann and circle cohomology,
and compute compatible superderivation algebras, particularly for
$N=1,2,3$. Ragoucy and Sorba \cite{RS92} extend this analysis to a general
compact smooth base, consider twists and noncompact extensions, and apply
it to superconformal, area-preserving and Krichever--Novikov algebras.
Their Sections~3--5 and Appendix~A already contain the closed-current and
compatibility mechanisms used in Sections~\ref{s:central} and
\ref{s:derivations} above. Explicit harmonic bases implement these mechanisms
in coordinates adapted to a group action; they do not replace their
cohomological classification.

The necessary signs can be stated without choosing Grassmann derivative
conventions. For homogeneous $f,g,h\in A_N$, a homogeneous cyclic cocycle
$\Phi$ obeys
\begin{align}\label{e:supercyclic}
 \Phi(f,g)&=-(-1)^{|f||g|}\Phi(g,f),\notag\\
 (-1)^{|f||h|}\Phi(fg,h)&+(-1)^{|g||f|}\Phi(gh,f)
 +(-1)^{|h||g|}\Phi(hf,g)=0.
\end{align}
Here $|f|\in\Z/2\Z$ is the parity of a homogeneous coefficient function; all parity sums are modulo two.
Then $\omega(xf,yg)=\kappa(x,y)\Phi(f,g)$ is a super Lie cocycle.
If $\Phi$ has parity $\epsilon$, its central generator has parity
$\epsilon$; the scalar value vanishes unless $|f|+|g|=\epsilon$.
For example, if $D$ is an even derivation and a Berezin trace $\tau$
satisfies $\tau(Df)=0$, then $\Phi(f,g)=\tau(fDg)$ is such a cocycle.
The trace has parity $N\bmod2$, so restricting in advance to even centers
would discard part of the general construction.

For a homogeneous superderivation $V$ fixing that center, the compatibility
condition is
\begin{equation}\label{e:supercompat}
 \Phi(Vf,g)+(-1)^{|V||f|}\Phi(f,Vg)=0.
\end{equation}
Berezin integration converts this into the graded counterpart of
\eqref{e:compat}; keeping \eqref{e:supercompat} as the primary equation
avoids ambiguities between left and right odd derivatives.
The analysis of \cite[Sec.~4.1.2]{CFRS90} shows that, for the standard
action on the same $N$-supercircle and a fixed central generator, the usual
$N$-extended superconformal algebra preserves a nonzero current cocycle
only for $N=1$ among the supersymmetric cases. This is not a nonexistence
theorem for $N>1$ superconformal algebras. Enlarging the superspace or
changing the extension problem leads to other compatible algebras; these
possibilities are part of \cite{CFRS90,RS92}.

There is a similarly instructive qualification in the area-preserving case.
On an oriented symplectic surface $(\Sigma,\mu)$, define $X_F$ by
$\iota_{X_F}\mu=\dd F$. It is divergence-free and determines a current
cocycle, even though the defining one-form is exact. If $V$ preserves
$\mu$, compatibility with this \emph{fixed} cocycle requires
$[V,X_F]=0$. Indeed, $\iota_{[V,X_F]}\mu=\cL_V(\iota_{X_F}\mu)-\iota_{X_F}(\cL_V\mu)=\cL_V\dd F=\dd(VF)$. Nondegeneracy of $\mu$ gives the equivalence with $\dd(VF)=0$; $VF$ is constant on each connected component. Thus the full area-preserving
algebra does not preserve every individual cocycle. The choice of a
cocycle, a family of cocycles or a transforming central space is essential
when relating current algebras to membrane symmetries.

\section{Castellani's loop geometry and the recent constructions}\label{s:castellani}
\subsection{What the three historical papers establish}
In \cite{Castellani88}, closed loops $X^\mu(\sigma)$ are expanded in
Fourier modes $x_n^\mu$. Castellani distinguishes arbitrary
\emph{loop-space diffeomorphisms}, whose parameters depend on the whole
loop, from reparametrizations of its circle parameter. A current
Poincar\'e algebra together with circle vector fields provides the algebra
to be gauged. In the consistent convention of \eqref{e:affine}, its
centerless kinematics is
\begin{equation}\label{e:Castloop}
 (\mathfrak{iso}(D-1,1)\otimes\C[t,t^{-1}])\rtimes\mathrm{Witt},
 \qquad [\ell_m,x_n]=-n x_{m+n}.
\end{equation}
Soft vielbeins, curvatures and horizontality conditions turn this algebraic
input into geometric field data. The paper omits central terms at the
classical starting point and discusses them separately in connection with
quantum BRST nilpotency.

The closing discussion of \cite[p.~50]{Castellani88} explicitly proposes
maps $G'/H\to G$ instead of $S^1\to G$, with generators indexed by
coset harmonics and brackets involving Clebsch--Gordan coefficients;
$S^2$ is given as a membrane example. It is therefore a direct conceptual
antecedent of \eqref{e:products}--\eqref{e:modejac}. The recent programme
makes particular harmonic realizations, spectral decompositions and
central/derivation data explicit. The general idea of replacing Fourier
modes by homogeneous-space harmonics should not be attributed exclusively
to the recent literature. Related Kaluza--Klein symmetry considerations
also predate it \cite{DolanDuff,SalamStrathdee}.

The 1989 paper \cite{Castellani89} goes beyond ordinary Lie algebras by
including higher-degree form generators. It constructs the loop extension
of an FDA containing a two-form in a representation $R$, together with
circle reparametrizations, and computes its BRST operator and the
normal-ordering constraints. For example, in that paper's normalization
the shifted quantum parameters satisfy
\begin{equation}\label{e:CastBRST}
 k'=-C_{\rm Adj}+2C_R,\qquad
 c'=26+2(\dim\mathfrak a-2\dim R)
 \quad\text{\cite[Eqs.~(3.13)]{Castellani89}}.
\end{equation}
Here $C_{\rm Adj}$ and $C_R$ are its trace coefficients for the finite
algebra $\mathfrak a$ and $R$. These representation- and ghost-dependent
quantum conditions are not consequences of a classical current
two-cocycle, and cannot be transplanted to a higher-dimensional base
without a new quantization calculation.

The 1990 paper \cite{Castellani90} develops the $D=4$, $N=1$
super-Poincar\'e version of this gauging. The fields are mode-indexed
vielbeins, spin connections, gravitini and a reparametrization gauge field.
In the rheonomic parametrization, curvature components containing fermionic one-forms (outer components) are constrained to be functions, typically linear combinations, of the bosonic vielbein components (inner components). These superspace constraints prevent extra independent fields in the fermionic directions. In \revred{Castellani's} loop version, imposing them together with the Bianchi identities yields analogues of the Einstein and gravitino equations \cite{Castellani90,CDF}.
An expansion in internal loop coordinates produces a tower of higher-spin
fields; the paper discusses mass generation through those coordinates
and a nonzero background for the zero-mode reparametrization field.
These are dynamical ingredients absent from the bare current bracket.
In particular, Bianchi identities for arbitrary unconstrained connections
are identities, not equations of motion: the rheonomic constraints carry
the extra content.

Castellani's conclusion \cite[Sec.~5]{Castellani90} explicitly distinguishes
this local loop-space field theory from the nonlocal loop-space
interactions associated with string joining and splitting. It does not
establish equivalence with a perturbative superstring theory. Thus the
historical title should not be read as a completed derivation of string
dynamics from an affine algebra.

\subsection{A precise algebraic bridge through higher cocycles}
An FDA is a free graded-commutative algebra with a degree-one differential
$q$, $q^2=0$. Its form grading is independent of super parity and of
Grassmann coefficient coordinates. For a Lie algebra $\mathfrak a$,
a representation $(R,\rho)$, and a closed CE three-cochain
$h\in Z^3_\CE(\mathfrak a,R)$, take anticommuting degree-one $a^a$ and
commuting degree-two $b^i$:
\begin{align}\label{e:FDA}
 qa^a&=-\tfrac12 f_{bc}{}^a a^ba^c,\notag\\
 qb^i&=-(\rho_a)^i{}_j a^ab^j-\tfrac16h^i{}_{abc}a^aa^ba^c.
\end{align}
The graded Leibniz rule is
$q(uv)=(qu)v+(-1)^{\deg u}u(qv)$; in particular
$q(a^ab^i)=(qa^a)b^i-a^a(qb^i)$. Independent super parity adds its own
signs. The equation $q^2=0$ gives Jacobi,
$[\rho_a,\rho_b]=f_{ab}{}^c\rho_c$ and $d_\CE h=0$.
A coboundary change in $h$ is absorbed by a quadratic shift of $b$.
More generally, a degree-$p$ form uses a degree-$(p+1)$ CE cocycle
\cite{Sullivan,DAuriaFre,Castellani89,CDF}; an ordinary scalar central
extension uses $H^2_\CE$ and a degree-one dual generator.

For commutative $A$, equip $R\otimes A$ with
$(xf)\cdot(rg)=\rho(x)r\otimes fg$. Then
\begin{equation}\label{e:hA}
 h_A(xf,yg,zk)=h(x,y,z)\otimes fgk
\end{equation}
is a three-cocycle: every term of $d_\CE h_A$ has the same product of
four functions, leaving $d_\CE h=0$. This extends loop FDA coefficients;
an action and quantum BRST analysis remain separate.

An $L_\infty$ algebra is a graded space $\mathfrak l$ with graded-skew
maps $l_k:\mathfrak l^{\otimes k}\to\mathfrak l$ of degree $2-k$ obeying
\begin{equation}\label{e:higherjacobi}
 \sum_{i+j=n+1}(-1)^{i(j-1)}\sum_{\sigma\in\operatorname{Sh}(i,n-i)}
 \chi(\sigma;x)\,l_j\bigl(l_i(x_{\sigma(1)},\ldots,x_{\sigma(i)}),
 x_{\sigma(i+1)},\ldots,x_{\sigma(n)}\bigr)=0.
\end{equation}
Here $n\ge1$, shuffles preserve order within each block, and $\chi$ is the
permutation sign times the graded exchange sign. The first identities give
$l_1^2=0$, Leibniz for $l_1,l_2$, and a Jacobi defect controlled by $l_3$;
only $l_2$ nonzero gives a graded Lie algebra \cite{LadaStasheff}.
For homogeneous $x_i$ and degree-zero coefficients, set
\begin{equation}\label{e:linfty}
 l_k^A(x_1f_1,\ldots,x_kf_k)
 =l_k(x_1,\ldots,x_k)\otimes f_1\cdots f_k.
\end{equation}
Every nested term has coefficient $f_1\cdots f_n$, with no extra signs,
so the higher identities hold, including those dual to an FDA.
In a harmonic algebra the cubic coefficients are
\begin{equation}\label{e:tripleC}
 \psi_I\psi_J\psi_K=\sum_P C_{IJK}^{P}\psi_P,\qquad
 C_{IJK}^{P}=\sum_Q C_{IJ}^{Q}C_{QK}^{P}.
\end{equation}
For Laurent modes this is $m_1+m_2+m_3=m$.
The tensor construction uses finite sums; the dual mode FDA requires a
completed dual, and smooth extensions require continuity and convergence.

Coefficient extension and adjoining higher generators can thus be combined.
For example, $h(x,y,z)=\kappa(x,[y,z])$ with trivial scalar coefficients
gives $qb=-h_{abc}a^aa^ba^c/6$. Its skeletal dual has $\mathfrak a$ in
degree zero, a central line in degree $-1$, $l_1=0$ and $l_3=h$.
The ordinary current two-cocycle instead adds a degree-zero central Lie
generator. These degrees differ; transgression, such as loop integration,
needs a specified degree-lowering map, not merely a common $\kappa$.

\subsection{\texorpdfstring{Current algebras of order $F$: central terms and an obstruction}{Current algebras of order F: central terms and an obstruction}}
\label{s:orderF}
Lie algebras of order $F$ provide another higher structure
\cite{RSorderF02,CRorder308,CRcolor09}, distinct from the graded-skew
$L_\infty$ operations above and from alternating Filippov brackets.
In the elementary $\Z/F\Z$-graded case, $F\ge2$, $\g_0$ is a Lie algebra, $\g_1$ a
$\g_0$-module, and $\mu_F:S^F\g_1\to\g_0$ a symmetric equivariant map.
With $[Y,X]=-[X,Y]$, the additional identities are
\begin{align}\label{e:Fidentities}
 [X,\mu_F(Y_1,\ldots,Y_F)]
 &=\sum_{j=1}^F\mu_F(Y_1,\ldots,[X,Y_j],\ldots,Y_F),\notag\\
 \sum_{j=1}^{F+1}[Y_j,\mu_F(Y_1,\ldots,\widehat Y_j,\ldots,Y_{F+1})]&=0.
\end{align}
The hat means omission. These are the equivariance and generalized Jacobi
identities of \cite[App.~A]{CRorder308}; $F=2$ recovers a Lie superalgebra.
The general definition allows further homogeneous components; we restrict
to this elementary case and a degree-zero scalar center.

For a commutative associative $A$, put $B_i=\g_i\otimes A$. Extend the
binary bracket and module action by multiplication, and set
\begin{equation}\label{e:Fcurrent}
 \mu_F^A(Y_1f_1,\ldots,Y_Ff_F)
 =\mu_F(Y_1,\ldots,Y_F)\otimes f_1\cdots f_F.
\end{equation}
Every nested identity has the same coefficient product, proving all
identities directly; $1\otimes D$, $D\in\Der(A)$, preserves every operation
by Leibniz. This uses finite tensor sums; smooth completions require the
continuity hypotheses already specified in Section~\ref{s:functions}.

Adjoin $\C z$ to $B_0$, let $z$ act trivially on $B_1$, and keep the
module action unchanged. The possible corrections are an alternating
$\omega:B_0\times B_0\to\C$ in the binary bracket and a symmetric
$\beta:S^FB_1\to\C$ in $\mu_F^A$, both multiplying $z$.
The extended identities hold if and only if $d_\CE\omega=0$ and
\begin{equation}\label{e:Fcentral}
 \omega(X,\mu_F^A(Y_1,\ldots,Y_F))
 =\sum_{j=1}^F\beta(Y_1,\ldots,[X,Y_j],\ldots,Y_F),
 \quad X\in B_0.
\end{equation}
Indeed, ordinary Jacobi gives the CE condition; central components of
equivariance give \eqref{e:Fcentral}; generalized Jacobi adds no condition
since $[Y,z]=0$. A section change $s'(X)=s(X)+\alpha(X)z$ gives
$\omega'=\omega-\alpha\circ[,]$ and
$\beta'=\beta-\alpha\circ\mu_F^A$. Thus the extension data are a coupled
pair, not merely an ordinary Lie two-cocycle. In particular $\beta=0$
requires $\omega(B_0,\operatorname{im}\mu_F^A)=0$.

A direct positive construction uses an invariant symmetric scalar
form $b\in(S^F\g_1^*)^{\g_0}$ and a linear $\tau:A\to\C$:
\begin{equation}\label{e:Fpositive}
 \omega=0,\qquad
 \beta(Y_1f_1,\ldots,Y_Ff_F)=b(Y_1,\ldots,Y_F)\tau(f_1\cdots f_F).
\end{equation}
For $X=xa$, the right-hand side of \eqref{e:Fcentral} factors as
$\tau(af_1\cdots f_F)$ times the invariance identity of $b$, hence vanishes.
If $\tau\circ D=0$, $1\otimes D$ also fixes this center.
For example, take $F=3$, $\g_0=\mathfrak{sl}_3(\C)$, and $\g_1$ its
adjoint copy, with the cubic product induced by $\kappa$
\cite[Eq.~(3.4)]{RSorderF02}:
\begin{equation}\label{e:Fcubic}
 \mu_3(\bar x,\bar y,\bar w)
 =\kappa(x,y)w+\kappa(x,w)y+\kappa(y,w)x.
\end{equation}
Bars identify the adjoint copy. The polarization of $\tr(y^3)$ gives
a nonzero invariant $b$. With unital $A$ and nonzero $\tau$, this produces
a nonzero $\beta$; it is not removed by a section change, since $B_0$ is
perfect and $\omega'=0$ would force $\alpha=0$.

An affine two-cocycle can nevertheless be obstructed even when $\beta$
is allowed. Use \eqref{e:Fcubic}, $A=\C[t,t^{-1}]$, $D=t\partial_t$,
$\tau(t^m)=\delta_{m,0}$, and
$\omega_k(xf,yg)=k\kappa(x,y)\tau(fDg)$. Choose
$q=\kappa(h,h)\ne0$, $X=ht^{-1}$ and
$(Y_1,Y_2,Y_3)=(\bar h t,\bar h,\bar h)$. All $[X,Y_j]$ vanish, whereas
\begin{equation}\label{e:Fobstruction}
 \omega_k(X,\mu_3^A(Y_1,Y_2,Y_3))
 =\omega_k(ht^{-1},3qht)=3kq^2.
\end{equation}
Thus no $\beta$ satisfies \eqref{e:Fcentral} for $k\ne0$.
This obstruction concerns the specified degree-zero extension with
unchanged module action; other types of extension and contracted
kinematical examples in \cite{CRorder308} require separate checks.

\Needspace{19\baselineskip}
\subsection{Comparison and limits of the correspondence}
\begin{center}
\small
\setlength{\tabcolsep}{4pt}
\begin{tabular}{@{}>{\raggedright\arraybackslash}p{.19\textwidth}>{\raggedright\arraybackslash}p{.365\textwidth}>{\raggedright\arraybackslash}p{.385\textwidth}@{}}
\toprule
Feature & Castellani, 1988--1990 & Recent harmonic constructions\\
\midrule
Basic variable & Loops and fields on (super)loop space & Functions on a prescribed manifold $M$\\\addlinespace[6pt]
Coefficient algebra & Circle Fourier modes; coset-harmonic extension proposed & Explicit compact, noncompact and soft-manifold realizations\\\addlinespace[6pt]
Finite algebra & Poincar\'e, super-Poincar\'e and FDA data & Primarily a finite-dimensional simple Lie algebra\\\addlinespace[6pt]
Extensions & Loop/FDA and reparametrization symmetries; quantum BRST analysis & Selected current cocycles and compatible derivations\\\addlinespace[6pt]
Softness & Dynamical vielbeins, curvatures and horizontality & Geometric measure, basis and operator data\\\addlinespace[6pt]
Dynamics & Constrained Bianchi identities and a proposed higher-spin field theory & Algebraic and representation-theoretic input; no general field equations\\
\addlinespace[6pt]
\bottomrule
\end{tabular}
\end{center}
The table compares \cite{Castellani88,Castellani89,Castellani90}
with \cite{CMRnoncompact,CMRnoncompactproc,CMRsoft,CMRreview}; neither
programme subsumes the other. Poincar\'e algebras are nonsemisimple with
degenerate Killing forms, so Theorem~\ref{t:universal} requires replacement
by an analysis of their additional cohomology \cite{NeebWagemann}.

Our finite-dimensional base $M$ supports coefficient functions. Castellani's
loop space has whole maps $X:S^1\to\text{spacetime}$ as points, with
Fourier coefficients as coordinates for its gauge fields; changing a metric
on $M$ does not introduce those fields. Constants form a Lie subalgebra,
but projection $P_0$ onto Laurent zero modes is generally not a Lie map:
$P_0([xt,yt^{-1}])=[x,y]$, while $[P_0(xt),P_0(yt^{-1})]=0$.
Discarded modes can therefore source retained ones, and a dynamical
truncation requires a separate consistency argument.

\section{Conclusions and research directions}\label{s:outlook}
The common mathematical structure is an algebra of coefficients, an
invariant form and a closed current. Harmonic analysis makes its brackets
explicit, while compatibility selects the derivations of a chosen
extension. The work of \cite{CFRS90,RS92} provides a cohomological
foundation; the recent compact and noncompact constructions provide
concrete realizations; Castellani's work adds the separate questions of
higher forms, gauging and dynamics. Equations \eqref{e:hA} and
\eqref{e:linfty} give a precise algebraic route for combining these inputs.
For order-$F$ currents, \eqref{e:Fcurrent} gives the coefficient
extension, while \eqref{e:Fcentral} is the extra test for scalar central terms.

Several problems remain distinct. First, classify the selected cocycles
under diffeomorphisms and determine their stabilizers, including the
effect of zeros or singular supports. Second, establish representation
theory with specified domains, positivity and regularization, rather than
infer an affine highest-weight theory from a mode basis. Third, distinguish
changes of harmonic realization from genuine deformations of an algebra
with specified topology and geometric data. Fourth, integrate Lie algebra
cocycles to current-group extensions: period and descent obstructions
remain, even after Jacobi is proved \cite{Neeb02,JanssensWockel}.
Finally, a higher-dimensional gauged FDA needs an action or rheonomic
constraints and a separately verified quantum consistency condition.

These distinctions also delimit physical applications. Harmonic products
organize Kaluza--Klein couplings, but a finite set of modes is usually not
closed under multiplication, so an arbitrary harmonic cutoff is not a
consistent truncation. Closure and dynamical consistency must be checked
in a specified model; generalized parallelizations furnish more structured
examples \cite{Blair24}. Likewise a central cocycle alone is not a proof
of anomaly cancellation, flux quantization, a supergravity duality or
holographic equivalence. The framework gives concrete algebraic data with
which such questions can be investigated.

\subsection*{Funding}
R.C.-S. has been partially supported by the Agencia Estatal de
Investigaci\'on (Spain) under grants PID2023-148373NB-I00, funded by
MCIN/AEI (\nolinkurl{10.13039/501100011033}) and FEDER, UE, and
PID2024-156578NB-I00, funded by MICIU/AEI
(\nolinkurl{10.13039/501100011033}) and FEDER, EU.

\bibliographystyle{amsplain}
\bibliography{ref}
\end{document}